\documentclass[english]{article}
\usepackage[T1]{fontenc}
\usepackage[latin9]{inputenc}
\usepackage{amsmath}
\usepackage{amsthm}
\usepackage{amssymb}

\makeatletter
  \theoremstyle{definition}
  \newtheorem{defn}{\protect\definitionname}
 \theoremstyle{definition}
 \newtheorem*{defn*}{\protect\definitionname}
\theoremstyle{plain}
\newtheorem{thm}{\protect\theoremname}
  \theoremstyle{plain}
  \newtheorem{assumption}{\protect\assumptionname}
  \theoremstyle{plain}
  \newtheorem{lem}{\protect\lemmaname}
  \theoremstyle{plain}
  \newtheorem{cor}{\protect\corollaryname}

\usepackage{babel}

\makeatother

\usepackage{babel}
  \providecommand{\assumptionname}{Assumption}
  \providecommand{\definitionname}{Definition}
  \providecommand{\lemmaname}{Lemma}
\providecommand{\corollaryname}{Corollary}
\providecommand{\theoremname}{Theorem}

\begin{document}

\title{Closed-form solutions for worst-case law invariant risk measures
with application to robust portfolio optimization}

\author{Jonathan Yu-Meng Li \footnote{Telfer School of Management, University of Ottawa, Ottawa, ON, Canada.\newline Contact: jonathan.li@telfer.uottawa.ca}} 
\maketitle
\begin{abstract}
Worst-case risk measures refer to the calculation of the largest value
for risk measures when only partial information of the underlying
distribution is available. For the popular risk measures such as Value-at-Risk
(VaR) and Conditional Value-at-Risk (CVaR), it is now known that their
worst-case counterparts can be evaluated in closed form when only
the first two moments are known for the underlying distribution. These
results are remarkable since they not only simplify the use of worst-case
risk measures but also provide great insight into the connection between
the worst-case risk measures and existing risk measures. We show in
this paper that somewhat surprisingly similar closed-form solutions
also exist for the general class of law invariant coherent risk measures,
which consists of spectral risk measures as special cases that are
arguably the most important extensions of CVaR. We shed light on the
one-to-one correspondence between a worst-case law invariant risk
measure and a worst-case CVaR (and a worst-case VaR), which enables
one to carry over the development of worst-case VaR in the context
of portfolio optimization to the worst-case law invariant risk measures
immediately. 
\end{abstract}

\section{Introduction}

Measuring how risky a random loss is often requires the knowledge
of its probability distribution. The industry standard measure of
risk, Value-at-Risk (VaR), for example, reports the risk level of
a random loss by calculating an extremal quantile of its distribution.
Another measure of risk, Conditional Value-at-Risk (CVaR), which has
emerged as the most popular alternative to replace VaR as industry
standard, calculates the average loss exceeding an extremal quantile
to indicate the riskiness of a random loss. The problem however of
implementing both of these measures and any other distribution-based
risk measure is that in most practices the exact form of distribution
is often lacking and only sample data is available for estimating
the distribution, which is inevitably prone to sampling error.

This issue has motivated the development of worst-case risk measures
where the goal is to determine the worst-possible risk level over
a set of candidate distributions that captures the uncertainty of
distribution. Worst-Case Value-at-Risk (WCVaR) was first studied by
El Ghaoui et al. (2003) \cite{Laurent-El-Ghaoui:2003aa}, who considered
a set of candidate distributions described by the first two moments,
and showed how the worst-possible VaR value can be calculated for
the set. One of the most noticeable results of El Ghaoui et al. (2003)
\cite{Laurent-El-Ghaoui:2003aa} is perhaps the closed-form solution
for WCVaR. The closed-form expression remarkably resembles the risk
measure of weighted mean-standard deviation, and hence provides useful
insight into how WCVaR can be minimized. El Ghaoui et al. provided
also the formulations of semidefinite programs that are equivalent
to the closed-form expression, which are useful when an additional
layer of uncertainty about the moments needs to be further addressed.
It turns out that a closed-form expression also exists for Worst-Case
Conditional Value-at-Risk (WCCVaR) when the set of candidate distributions
is described by the first two moments (see Chen et al. (2011)\cite{Chen:2011aa},
Natarajan et al. (2010)\cite{Natarajan:2010aa}), and the expression
is identical to the one for WCVaR. Interestingly, this implies that
some of the developments in WCVaR such as dealing with moment uncertainty
in El Ghaoui et al. (2003) can be directly carried over to the case
of WCCVaR. Alternative formulations of worst-case risk measures can
also be found in the literature of distributionally robust optimization
(DRO) (see for example \cite{Calafiore:2007aa,Delage:2010aa,Natarajan:2010aa,Wiesemann:2014aa,Wozabal:2014aa,Zhu:2009aa}).
Most of these works focus on deriving tractable convex or conic programs
for computing the worst-case values (and finding the corresponding
robust solutions).

Our work is motivated by the insight gained from the closed-form solutions
of the WCVaR and WCCVaR. Given the elegancy of the closed form, it
is natural to wonder if the closed-form result is just a consequence
of the relatively simple structure of VaR and CVaR, or it can be found
also for alternative risk measures with more sophisticated structure.
On the top of the list of more sophisticated risk measures is the
class of spectral risk measures that plays an essential role in both
theory and practice. They were first introduced by Acerbi (2002) \cite{Acerbi:smmrcrsra}
who attempted to generalize CVaR (and VaR) so that a more realistic
description of risk-aversion can be made over a spectrum of CVaRs
(VaRs). Later, it became clear that this class of measures is equivalent
to the class of distortion risk measures that have applications in
insurance \cite{Pflug:2006aa,Pichler:2015aa}. It is also known that
spectral risk measures satisfy most, if not all, desirable properties
that have been postulated by the modern risk theory (\cite{Acerbi:smmrcrsra,artzner:coherentRM,follmer02:cvxRiskMeas,kusuoka:licrm}),
namely the property of monotonicity, convexity, translation invariance,
coherency, and law invariance. A more surprising finding however is
that any risk measure that satisfies all these properties, also known
as law invariant coherent risk measures, can be represented through
spectral risk measures (see \cite{kusuoka:licrm,Shapiro:2013aa}).
We study in this paper both the case of Worst-Case Spectral Risk Measure
(WCSRM) and Worst-Case Law-Invariant Coherent Risk Measure (WCLICRM).
Our finding is that despite their seeming complexity, both can be
boiled down to a closed-form expression when only the first two moments
are known for the underlying distribution. The closed-form remarkably
resembles the measure of weighted mean-standard deviation also, which
allows us to shed light on the one-to-one correspondence between any
WCLICRM and WCCVaR (and WCVaR). Based on the observation, we demonstrate
how the result can be extended and applied in the context of robust
portfolio optimization.

This article is organized as follows. In Section 2, we prove the closed-form
result for WCSRM and WCLICRM over a set of univariate distributions
with fixed first two moments. We show in Section 3 how the result
can be applied in the context of robust portfolio optimization.

\section{Analytical Results}

Let $(\Omega,{\cal F},\mathbb{P})$ be a probability space and $Z$
denote a random variable with its distribution $F_{Z}$, i.e. $Z:(\Omega,{\cal F},\mathbb{P})\rightarrow\Re$
and $F_{Z}(t):=\mathbb{P}(Z\leq t)$. The space of random variables
is contained in $L^{2}(\Omega,{\cal F},\mathbb{P})$. We begin by
recalling the following definition of spectral risk measure. 
\begin{defn}
(Spectral risk measure \cite{Acerbi:smmrcrsra}) Given a random variable
$Z$, let $F_{Z}^{-1}$ denote its general inverse cdf function, i.e.
$F_{Z}^{-1}(\alpha):=\inf\{q\;|\;F_{Z}(q)\geq\alpha\}$. The function
\begin{align*}
\rho_{\phi}(Z):= & \int_{0}^{1}\phi(\alpha)F_{Z}^{-1}(\alpha)d\alpha
\end{align*}
is called a spectral risk measure parameterized by $\phi$, if $\phi\in L^{1}[0,1]$
is a non-decreasing probability density function, i.e. $\phi\geq0$
and $\int_{0}^{1}\phi(\alpha)d\alpha=1$. The density function $\phi$
is also called an ``admissible'' risk spectrum. 
\end{defn}
Intuitively, a spectral risk measure may be viewed as a weighted sum
of Value-at-Risk (VaR), where the admissibility of $\phi$ enforces
that the weight assigned to a larger VaR cannot be less. This characterizes
the coherency required for a rational individual who is risk-averse.
The most notable example of spectral risk measure is $(1-\epsilon)-$Conditional
Value-at-Risk ($(1-\epsilon)-$CVaR), where the spectrum $\phi$ takes
the form $\phi(\alpha):=\frac{1}{\epsilon}\mathbf{1}_{[1-\epsilon,1)}(\alpha)$
and $\epsilon\in(0,1]$ stands for a tail probability of $F_{Z}$.
To understand why spectral risk measures play an central role in the
modern theory of risk measures \cite{artzner:coherentRM,follmer02:cvxRiskMeas,kusuoka:licrm},
we shall review the following definition about law invariant coherent
risk measures. 
\begin{defn*}
(Law invariant coherent risk measures) Any risk measure $\rho:L^{2}(\Omega,{\cal F},\mathbb{P})\rightarrow\Re$
that satisfies

1) Monotonicity: $\rho(Z_{1})\leq\rho(Z_{2})$ for any $Z_{1}\leq Z_{2}$
almost surely;

2) Convexity : $\rho((1-\lambda)Z_{1}+\lambda Z_{2})\leq(1-\lambda)\rho(Z_{1})+\lambda\rho(Z_{2}),\;0\leq\lambda\leq1$;

3) Translation invariance : $\rho(Z+c)=\rho(Z)+c$, $c\in\Re$;

4) Positive homogeneity : $\rho(\lambda Z)=\lambda\rho(Z)$, $\lambda\geq0$;

5) Law invariance : $\rho(Z_{1})=\rho(Z_{2})$ if $F_{Z_{1}}\equiv F_{Z_{2}}$.
\\
is said to be a law invariant coherent risk measure. 
\end{defn*}
The importance of the above class of risk measures lies in the fact
that it satisfies all the properties that have been postulated by
the modern theory of convex risk measures \cite{artzner:coherentRM,follmer02:cvxRiskMeas,kusuoka:licrm}
about what a reasonable risk measure should satisfy. Interestingly,
despite its generality, there is a close link between this general
class of risk measures and spectral risk measures, namely that the
former can always be represented through the latter via a supremum
representation. 
\begin{thm}
\label{Any-law-invariant}Any law invariant coherent risk measure
$\rho_{\Phi}:L^{2}(\Omega,{\cal F},\mathbb{P})\rightarrow\Re$ has
the following representation 
\[
\rho_{\Phi}(Z)=\sup_{\phi\in\Phi}\rho_{\phi}(Z),
\]
where $\Phi\subseteq L^{1}[0,1]$ denotes a set of admissible spectrums. 
\end{thm}
\begin{proof}
It has been discussed in \cite{Shapiro:2013aa} (see Proposition 1)
that any law invariant coherent risk measure $\rho:L^{p}(\Omega,{\cal F},\mathbb{P})\rightarrow\Re$
with $p\in[1,\infty)$ admits the representation of
\[
\rho(Z)=\sup_{\mu\in{\cal M}}\int_{0}^{1}AV@R_{\gamma}(Z)d\mu(\gamma),
\]
where $AV@R_{\gamma}$ stands for $\gamma-$CVaR and ${\cal M}$ denotes
some set of probability measures on $[0,1)$. Since $AV@R_{\gamma}(Z)=(1-\gamma)^{-1}\int_{\gamma}^{1}F_{Z}^{-1}(\alpha)d\alpha$,
we have
\begin{align*}
\text{\ensuremath{\rho}}(Z)= & \sup_{\mu\in{\cal M}}\int_{0}^{1}\int_{\gamma}^{1}(1-\gamma)^{-1}F_{Z}^{-1}(\alpha)d\alpha d\mu(\gamma)\\
= & \sup_{\phi\in\Phi}\int_{0}^{1}\phi(\alpha)F_{Z}^{-1}(\alpha)d\alpha=\sup_{\phi\in\Phi}\rho_{\phi}(Z),
\end{align*}
where $\Phi:=\{\varphi\;|\;\varphi(\kappa)=\int_{0}^{\kappa}(1-\gamma)^{-1}d\mu(\gamma),\;\kappa\in[0,1],\;\,\mu\in{\cal M}\}$
and every $\varphi$ is, by definition, a non-decreasing probability
density on $[0,1]$, i.e. it is an admissible spectrum. 
\end{proof}
As mentioned earlier, for both the case of VaR and CVaR, their worst-case
counterparts can be evaluated in closed form when only the first two
moments are known for the underlying distribution. More specifically,
given a pair of mean and standard deviation $(\mu,\sigma)$, the largest
$(1-\epsilon)-$VaR and $(1-\epsilon)-$CVaR value over the set of
distributions having the mean $\mu$ and standard deviation $\sigma$
can be calculated by \cite{Laurent-El-Ghaoui:2003aa,Chen:2011aa}

\begin{equation}
\rho_{WCVaR}(\mu,\sigma,\epsilon)=\rho_{WCCVaR}(\mu,\sigma,\epsilon)=\mu+\sigma\sqrt{\frac{1-\epsilon}{\epsilon}},\label{eq:wvar}
\end{equation}
where $\epsilon\in(0,1]$ is the tail probability.

Along this line of work, we consider the following optimization problem
that defines the Worst-Case Law Invariant Coherent Risk Measures (WCLICRM):
\begin{eqnarray}
\rho_{WCLICRM}(\mu,\sigma,\Phi):= & \sup_{F_{Z}\in\mathcal{Q}} & \rho_{\Phi}(Z)\label{eq:wclicrm}\\
 & {\rm subject}\;{\rm to} & \mathbb{E}[Z]=\mu\nonumber \\
 &  & \mathbb{STD}[Z]=\sigma,\nonumber 
\end{eqnarray}
where ${\cal Q}$ denotes the set of all probability distributions
on $(-\infty,\infty)$. As a special case of WCLICRM, we define also
the Worst-Case Spectral Risk Measures (WCSRM) when a single spectrum
$\phi$ is considered
\begin{equation}
\rho_{WCSRM}(\mu,\sigma,\phi):=\rho_{WCLICRM}(\mu,\sigma,\{\phi\}).\label{eq:wcsrm}
\end{equation}

Before proceeding further, we shall make the following assumption
about the risk measure $\rho_{\Phi}$ used in defining the problem
$\eqref{eq:wclicrm}$. 
\begin{assumption}
\label{assu:For-any-risk}For any risk measure $\rho_{\Phi}$ employed
in the definition of WCLICRM, the set $\Phi$ consists of spectrums
in $L^{\infty}[0,1]$, i.e. bounded functions only.
\end{assumption}
As noted in \cite{Pichler:2013aa}, unless all considered random variables
are essentially bounded, i.e. $Z\in L^{\infty}(\Omega,{\cal F},\mathbb{P})$,
in general a spectral risk measure $\rho_{\phi}$ with an arbitray
spectrum $\phi\in L^{1}[0,1]$ may not be well defined. It is not
hard to confirm that with Assumption \ref{assu:For-any-risk}, a spectral
risk measure $\rho_{\phi}$ would be finite for any $Z\in L^{2}(\Omega,{\cal F},\mathbb{P})$
(in fact, for any $Z\in L^{1}(\Omega,{\cal F},\mathbb{P}))$. Moreover,
this assumption is not really restrictive since for any general law
invariant coherent risk measure $\rho$, there always exists a set
$\Phi\subseteq L^{\infty}[0,1]$ such that $\rho=\sup_{\phi\in\Phi}\rho_{\phi}$
holds (see Corollary 5, \cite{Pflug:2015aa}).

It is not clear if the problem (\ref{eq:wclicrm}) is tractable in
its full generality or only for special cases like CVaR. The main
result of this section is to show that not only can the above problem
be tractably solved for the cases where $\Phi:=\{\phi\}$, i.e. the
case of WCSRM (\ref{eq:wcsrm}), the solution of (\ref{eq:wclicrm})
in general admits an elegant closed form expression. We present the
result in two steps. Firstly we focus on the case of spectral risk
measures, i.e. (\ref{eq:wcsrm}) and show that in this case the problem
(\ref{eq:wcsrm}) reduces to a closed-form. Thereafter, the result
of WCLICRM, i.e. (\ref{eq:wclicrm}) in general, can be proved fairly
straightforwardly.

Before presenting our main results, we need the following lemma that
facilitates our analysis. 
\begin{lem}
\label{Any-spectral-risk} Given Assumption \ref{assu:For-any-risk},
any spectral risk measure $\rho_{\phi}(Z)$ can be equivalently formulated
as 
\begin{align*}
\rho_{\phi}(Z)= & \min_{q(\alpha)}\mathbb{E}[\phi(0)Z+\int_{0}^{1}[(1-\alpha)q(\alpha)+(Z-q(\alpha))^{+}]d\phi(\alpha)],
\end{align*}
where $q\in L^{1}(0,1)$, and there exists a non-decreasing function
as the optimal solution. 
\end{lem}
\begin{proof}
Following Proposition 3.2 in \cite{Acerbi:2002aa} we have 
\begin{align*}
\rho(Z) & =\min_{q(\alpha)}\phi(0)\mathbb{E}[Z]+\int_{0}^{1}[(1-\alpha)q(\alpha)+\mathbb{E}[(Z-q(\alpha))^{+}]]d\phi(\alpha),
\end{align*}
and the optimal solution $q^{*}(\alpha)$ satisfies $q^{*}(\alpha)\in[F_{Z}^{-1}(\alpha),F_{Z}^{-1}(\alpha)^{+})$\footnote{$F_{Z}^{-1}(\alpha)^{+}:=\inf\{q\;|\;F_{Z}(q)>\alpha\}$}
over $\alpha\in\text{supp}(\phi)$, i.e. the support of the measure
defined by $\phi(\alpha)$, and can take arbitrary values otherwise.
Hence, one can always construct a non-decreasing function $q^{*}(\alpha)$
over $(0,1)$ that attains the optimality. 

Applying Fubini's Theorem, we arrive at the result. 
\end{proof}
Applying Lemma \ref{Any-spectral-risk}, we can equivalently formulated
the problem of WCSRM (\ref{eq:wcsrm}) as

\begin{align}
\rho_{WCSRM}(\mu,\sigma,\phi)=\sup_{F_{Z}} & \,\min_{q(\alpha)}\int_{-\infty}^{\infty}[\phi(0)z+\int_{0}^{1}[(1-\alpha)q(\alpha)+(z-q(\alpha))^{+}]d\phi(\alpha)]dF_{Z}\label{eq:srmfull}\\
{\rm subject}\;{\rm to} & \int_{-\infty}^{\infty}zdF_{Z}=\mu\nonumber \\
 & \int_{-\infty}^{\infty}z^{2}dF_{Z}=\mu^{2}+\sigma^{2}\nonumber \\
 & \int_{-\infty}^{\infty}dF_{Z}=1,\;F_{Z}\in\mathcal{Q}.\nonumber 
\end{align}

For simplicity, from here on the integral $\int_{-\infty}^{\infty}$
may be written as $\int$ only. As the main result of this paper,
we show in the following theorem that the above problem can be reduced
to the form of a weighted sum of mean and standard deviation. 
\begin{thm}
\label{Any-worst-case-spectral}Given Assumption 1, any worst-case
spectral risk measure (WCSRM) can be evaluated in closed-form: 
\begin{align*}
\rho_{WCSRM}(\mu,\sigma,\phi)= & \mu+\sigma\sqrt{\int_{0}^{1}\phi^{2}(p)dp-1}.
\end{align*}
In the case of $(1-\epsilon)-$CVaR, we have $\int_{0}^{1}\phi^{2}(p)dp=\frac{1}{\epsilon}$. 
\end{thm}
\begin{proof}
Firstly, given that in (\ref{eq:srmfull}) for any fixed $F_{Z}$,
there exists a non-decreasing function as the optimal solution, we
can impose without loss of generality the constraint $q(\alpha)\in Q^{\nearrow}$,
where $Q^{\nearrow}$ denotes the set of all non-decreasing functions
over $(0,1)$. This will facilitate the rest of the proof. Let also
$g(z;q(\alpha)):=\int_{0}^{1}[(1-\alpha)q(\alpha)+(z-q(\alpha))^{+}]d\phi(\alpha)$.

Since the set $Q^{\nearrow}$ is convex, and the objective function
in (\ref{eq:srmfull}) is convex in $q(\alpha)$ for any fixed $F_{Z}$
and linear in $F_{Z}$ for any fixed $q(\alpha)$, we can apply Sion's
minmax theorem \cite{Sion:1958aa} to switch the sup and min and arrive
at the following equivalent problem 
\begin{eqnarray*}
\min_{q(\alpha)\in Q^{\nearrow}} & \{\sup_{F_{Z}\in{\cal Q}} & \int[\phi(0)z+g(z;q(\alpha))]dF_{Z}\;|\;\int zdF_{Z}=\mu,\int z^{2}dF_{Z}=\mu^{2}+\sigma^{2}\}.
\end{eqnarray*}

Applying duality theory of conic linear problems (Shapiro 2001 \cite{Shapiro:2001aa}),
we can replace the inner maximization problem by its dual, which leads
to 
\begin{eqnarray}
\min_{q(\alpha)\in Q^{\nearrow}} & \min_{\lambda_{0},\lambda_{1},\lambda_{2}} & \lambda_{0}+\mu\lambda_{1}+(\mu^{2}+\sigma^{2})\lambda_{2}\label{eq:minmax}\\
 & {\rm subject}\;{\rm to}, & \lambda_{0}+z\lambda_{1}+z^{2}\lambda_{2}\geq\phi(0)z+g(z;q(\alpha)),\;\forall z.\label{eq:con1}
\end{eqnarray}

Strong duality holds when $\sigma>0$. It is easy to verify that when
$\sigma=0$, we have $\rho_{WCSRM}(\mu,0,\phi)=\mu$ and therefore
only the case $\sigma>0$ requires further investigation.

We claim that given any fixed $q(\alpha)\in Q^{\nearrow}$, the function
$g(z;q(\alpha))$ is equivalent to the following function 
\begin{align*}
h(z;q(\alpha)):=\sup_{\beta\in[0,1)}\int_{0}^{1}[(1-\alpha)q(\alpha)+{\bf 1}_{[0,\beta]}(\alpha)(z-q(\alpha))]d\phi(\alpha).
\end{align*}

We verify this by considering the following two cases for any $z$
value: based on the given $q(\alpha)$, either there exists $\alpha\in(0,1)$
such that $q(\alpha)\leq z$ or otherwise. For the first case, let
$\beta(z):=\arg\max_{\{\alpha\in(0,1)\}}\{\alpha\;:\;q(\alpha)\leq z\}$.
Since $q(\alpha)$ is non-decreasing, we have $q(\alpha)\leq z,\forall\alpha\in(0,\beta(z)]$.
Thus, we can equivalently re-write $g(z;q(\alpha))$ as 
\begin{align*}
g(z;q(\alpha))=\int_{0}^{1}[(1-\alpha)q(\alpha)+{\bf 1}_{(0,\beta(z)]}(\alpha)(z-q(\alpha))]d\phi(\alpha).
\end{align*}

By definition, $h(z;q(\alpha))\geq g(z;q(\alpha))$ follows. To show
the other direction, let $\beta^{*}(z)$ denote the optimal solution
for the problem in $h(z;q(\alpha))$. There are two possible cases:
either $\beta^{*}(z)\leq\beta(z)$ or otherwise. If $\beta^{*}(z)\leq\beta(z)$,
we have 
\begin{align*}
h(z;q(\alpha))= & \int_{0}^{1}[(1-\alpha)q(\alpha)+{\bf 1}_{(0,\beta^{*}(z)]}(\alpha)(z-q(\alpha))]d\phi(\alpha)\\
\leq & \int_{0}^{1}[(1-\alpha)q(\alpha)+{\bf 1}_{(0,\beta^{*}(z)]}(\alpha)(z-q(\alpha))+{\bf 1}_{(\beta^{*}(z),\beta(z)]}(\alpha)(z-q(\alpha))]d\phi(\alpha)\\
= & g(z;q(\alpha)),
\end{align*}
whereas for the case $\beta^{*}(z)>\beta(z)$ we have 
\begin{align*}
h(z;q(\alpha))= & \int_{0}^{1}[(1-\alpha)q(\alpha)+{\bf 1}_{(0,\beta(z)]}(\alpha)(z-q(\alpha))+{\bf 1}_{(\beta(z),\beta^{*}(z)]}(\alpha)(z-q(\alpha))]d\phi(\alpha)\\
\leq & g(z;q(\alpha)),
\end{align*}
where the last inequality is due to the definition of $\beta(z)$.

Now, for the case that there exists no $\alpha\in(0,1)$ such that
$q(\alpha)\leq z$, we immediately have 
\[
h(z;q(\alpha))=\int_{0}^{1}(1-\alpha)q(\alpha)d\phi(\alpha)=g(z;q(\alpha)),
\]
where the first equality is due to that the optimal $\beta^{*}(z)$
in $h(z;q(\alpha))$ must be zero, and the second one follows the
definition of $g(z;q(\alpha))$.

Hence, we can replace $g(z;q(\alpha))$ by $h(z,q(\alpha))$ in the
constraint (\ref{eq:con1}), which leads to
\[
\lambda_{0}+z\lambda_{1}+z^{2}\lambda_{2}\geq\phi(0)z+\int_{0}^{\beta}[(1-\alpha)q(\alpha)+(z-q(\alpha))]d\phi(\alpha)+\int_{\beta}^{1}[(1-\alpha)q(\alpha)]d\phi(\alpha),\;\forall z,\forall\beta\in[0,1],
\]
or 
\begin{align*}
\min_{z}\{(\lambda_{0}-\int_{0}^{1}[(1-\alpha)-{\bf 1}_{(0,\beta]}(\alpha)]q(\alpha)d\phi(\alpha))+(\lambda_{1}-(\phi(0)+\int_{0}^{\beta}d\phi(\alpha)))z+\lambda_{2}z^{2}\}\geq0,\;\forall\beta\in[0,1].
\end{align*}

By definition, $\phi(0)+\int_{0}^{\beta}d\phi(\alpha)=\phi(\beta)$.
For any fixed $\beta$, the left-hand-side of the above inequality
is an elementary minimization problem of a univariate quadratic function.
It is bounded below if any only if $\lambda_{2}\geq0$. By replacing
the optimization problem by the formula known for its optima value,
we have the following equivalent formulation: 
\begin{align*}
(\lambda_{0}-\varphi(\beta))-\frac{(\lambda_{1}-\phi(\beta))^{2}}{4\lambda_{2}}\geq0,\;\forall\beta\in[0,1],
\end{align*}
where $\lambda_{2}\geq0$ and $\varphi(\beta):=\int_{0}^{1}[(1-\alpha)-{\bf 1}_{(0,\beta]}(\alpha)]q(\alpha)d\phi(\alpha).$

The optimization problem (\ref{eq:minmax}) with the above reformulated
constraints can be further reformulated into 
\begin{align*}
 & \min_{q(\alpha)\in Q^{\nearrow},\lambda_{1},\lambda_{2}\geq0}\sup_{\beta\in[0,1]}\{\frac{(\lambda_{1}-\phi(\beta))^{2}}{4\lambda_{2}}+\varphi(\beta)+\mu\lambda_{1}+(\mu^{2}+\sigma^{2})\lambda_{2}\}\\
\Rightarrow & \min_{q(\alpha)\in Q^{\nearrow},\lambda_{1},\lambda_{2}\geq0}\sup_{\beta\in[0,1]}\{\frac{\lambda_{1}^{2}}{4\lambda_{2}}-\frac{\lambda_{1}\phi(\beta)}{2\lambda_{2}}+\frac{\phi(\beta)^{2}}{4\lambda_{2}}+\varphi(\beta)+\mu\lambda_{1}+(\mu^{2}+\sigma^{2})\lambda_{2}\}\\
\Rightarrow & \min_{q(\alpha)\in Q^{\nearrow},q,r\geq0}\sup_{\beta\in[0,1]}\{\phi(\beta)^{2}r+\phi(\beta)q+\varphi(\beta)\}+\frac{q^{2}}{4r}+\mu(\frac{-q}{2r})+(\mu^{2}+\sigma^{2})\frac{1}{4r}\\
\Rightarrow & \min_{q(\alpha)\in Q^{\nearrow},q,r\geq0}\sup_{\beta\in[0,1]}\{\phi(\beta)^{2}r+\phi(\beta)q+\varphi(\beta)\}+\frac{(q-\mu)^{2}+\sigma^{2}}{4r}
\end{align*}
where $r=\frac{1}{4\lambda_{2}}$ and $q=\frac{-\lambda_{1}}{2\lambda_{2}}$
is applied in the third line.

By introducing dummy variables $s,t\in\Re$, we have the following
equivalent formulation 
\begin{align}
\min_{q(\alpha)\in Q^{\nearrow},q,r,s,t} & s+\int_{0}^{1}(1-\alpha)q(\alpha)d\phi(\alpha)+t\nonumber \\
 & \int_{0}^{1}{\bf 1}_{(0,\beta]}(\alpha)q(\alpha)d\phi(\alpha)-\phi(\beta)^{2}r-\phi(\beta)q+s\geq0,\forall\beta\in[0,1]\label{eq:con}\\
 & 4rt\geq(q-\mu)^{2}+\sigma^{2}\nonumber \\
 & r\geq0,\nonumber 
\end{align}
where the second constraint can be recast as a second order cone constraint
\begin{equation}
\begin{pmatrix}q-\mu\\
\sigma\\
r-t\\
r+t
\end{pmatrix}\in{\cal Q}_{4},\label{eq:socp}
\end{equation}
where ${\cal Q}_{4}:=\{(u,t)\in\Re^{4}\;|\;||u||\leq t\}$(see, e.g.
\cite{Alizadeh:2003aa}). To further reduce the problem, we relax
first the constraint $q(\alpha)\in Q^{\nearrow}$ and will verify
later that the relaxation is tight. We apply again the theory of conic
linear program \cite{Shapiro:2001aa} and derive the dual of the relaxed
problem.

We can define the dual variable for (\ref{eq:con}) by $y(\beta)\in{\cal Y}[0,1]$,
where ${\cal Y}[0,1]$ denotes the set of right continuous functions
of bounded variation on $[0,1]$ that corresponds to the space of all
finite signed Borel measures on $[0,1]$. The integral over $y(\beta)$
follows Lebesgue-Stieltjes integral. In addition, let $y\in\Re^{4}$
denote the dual variables corresponding to the second order cone constraint
\eqref{eq:socp}. We can write the Lagrange function as follows, where
$x:=(q(\alpha),q,r,s,t)$, 
\begin{align*}
L(x,y(\beta),y)= & \int_{0}^{1}(1-\alpha)q(\alpha)d\phi(\alpha)+s+t-\int_{0}^{1}[\int_{0}^{\beta}q(\alpha)d\phi(\alpha)-\phi(\beta)^{2}r-\phi(\beta)q+s]dy(\beta)\\
 & -y_{0}(q-\mu)-y_{1}(\sigma)-y_{2}(r-t)-y_{3}(r+t)\\
= & \int_{0}^{1}(1-\alpha)q(\alpha)d\phi(\alpha)-\int_{0}^{1}[\int_{\alpha}^{1}dy(\beta)]q(\alpha)d\phi(\alpha)\\
 & +s(1-\int_{0}^{1}dy(\beta))+t(1+y_{2}-y_{3})+r(\int_{0}^{1}\phi(\beta)^{2}dy(\beta)-y_{2}-y_{3})\\
 & +q(\int_{0}^{1}\phi(\beta)dy(\beta)-y_{0})+y_{0}\mu-y_{1}\sigma,
\end{align*}
where in the second line of equality the second term is obtained by
interchanging the order of integration. The dual problem $\max_{y(\beta),y}\min_{x}L(x,y(\beta),y)$
reduces to the following problem 
\begin{align}
\max_{y(\beta),y} & \mu y_{0}-\sigma y_{1}\nonumber \\
\text{subject to} & (1-\alpha)-\int_{\alpha}^{1}dy(\beta)=0,\forall\alpha\in\text{supp}(\phi)\label{eq:d1}\\
 & \int_{0}^{1}dy(\beta)=1\label{eq:d2}\\
 & y\succeq_{+}0\label{eq:d0}\\
 & y_{2}-y_{3}=-1\label{eq:d3}\\
 & y_{2}+y_{3}\leq\int_{0}^{1}\phi(\beta)^{2}dy(\beta)\label{eq:d4}\\
 & y_{0}=\int_{0}^{1}\phi(\beta)dy(\beta)\label{eq:d5}\\
 & \sqrt{y_{0}^{2}+y_{1}^{2}+y_{2}^{2}}\leq y_{3},\label{eq:d7}
\end{align}
where $y\succeq_{+}0$ refers to that $y(\beta)$ is non-decreasing
on $[0,1]$ .

Following Shapiro (2001)\cite{Shapiro:2001aa}, strong duality holds
if there exists a feasible $(q^{*}(\alpha),q^{*},r^{*},s^{*},t^{*})$
such that (generalized) slater condition can be satisfied. This is
the case for (\ref{eq:con}), since given any feasible solution that
does not satisfy the slater condition, we can always find alternative
feasible $s$ and $t$ so that the condition can be satisfied.

Observe that from (\ref{eq:d7}) we have
\begin{align}
 & y_{1}^{2}\leq y_{3}^{2}-y_{2}^{2}-y_{0}^{2}\nonumber \\
\Rightarrow & y_{1}\geq-\sqrt{y_{3}^{2}-y_{2}^{2}-y_{0}^{2}}\nonumber \\
\Rightarrow & y_{1}\geq-\sqrt{(y_{3}-y_{2})(y_{3}+y_{2})-y_{0}^{2}}\nonumber \\
\Rightarrow & -\sigma y_{1}\leq\sigma\sqrt{\int_{0}^{1}\phi(\beta)^{2}dy(\beta)-(\int_{0}^{1}\phi(\beta)dy(\beta))^{2}}\;\;\;\text{(due to}\;\eqref{eq:d3},\eqref{eq:d4},\eqref{eq:d5})\label{eq:las}
\end{align}
Since $y_{1}$ is only constrained by the above inequality, the equality
must hold for the optimal solution. 

Observe also that by applying integration by parts, we can write \eqref{eq:d1}
as
\begin{align}
y(1)-y(\alpha)=1-\alpha & ,\forall\alpha\in\text{supp}(\phi).\label{eq:las3}
\end{align}
Without loss of generality, we can assume $y(\beta)$ is normalized
by $y(1)=1$. Together with \eqref{eq:d5} and \eqref{eq:las}, the
objective function can now be reformulated into
\begin{align}
 & \mu(\int_{0}^{1}\phi(\beta)dy(\beta))+\sigma\sqrt{\int_{0}^{1}\phi(\beta)^{2}dy(\beta)-(\int_{0}^{1}\phi(\beta)dy(\beta))^{2}}\label{eq:obj}\\
= & \mu(\int_{\text{supp}(\phi)}\phi(\beta)dy(\beta)+\int_{[0,1]\setminus\text{supp}(\phi)}\phi(\beta)dy(\beta))+\\
 & \sigma\sqrt{(\int_{\text{supp}(\phi)}\phi(\beta)^{2}dy(\beta)+\int_{[0,1]\setminus\text{supp}(\phi)}\phi(\beta)^{2}dy(\beta))-(\int_{\text{supp}(\phi)}\phi(\beta)dy(\beta)+\int_{[0,1]\setminus\text{supp}(\phi)}\phi(\beta)dy(\beta))^{2}}.\label{eq:las2}
\end{align}

Observe that the integrals $\int_{[0,1]\setminus\text{supp}(\phi)}$
can be carried out independently from the exact shape of $y(\beta)$
over $[0,1]\setminus\text{supp}(\phi)$. Hence, we can always find
an optimal $y(\beta)$ by setting $y(\beta)=\beta$ that satisfies
\eqref{eq:d2}, \eqref{eq:d0}, \eqref{eq:las3}, i.e. it corresponds
to a uniform measure over $[0,1]$. The objective function \eqref{eq:las2}
can thus be reduced to
\begin{align*}
 & \mu(\int_{0}^{1}\phi(\beta)d\beta)+\sigma\sqrt{\int_{0}^{1}\phi(\beta)^{2}d\beta-(\int_{0}^{1}\phi(\beta)d\beta)^{2}}\\
= & \mu+\sigma\sqrt{\int_{0}^{1}\phi(\beta)^{2}d\beta-1,}
\end{align*}
since $\int_{0}^{1}\phi(\beta)d\beta=1$. 

We are left to show that the problem (\ref{eq:con}) remains tight
after relaxing the constraint $q(\alpha)\in Q^{\nearrow}$. By Shapiro
(2001) \cite{Shapiro:2001aa}, given that strong duality holds, we
have the following complementary condition hold for the primal optimal
solution ($q^{*}(\alpha),q^{*},r^{*},s^{*},t^{*})$ and dual optimal
solution ($y^{*}(\beta),y^{*})$
\begin{align*}
 & \int_{0}^{1}[\int_{0}^{\beta}q^{*}(\alpha)d\phi(\alpha)-\phi(\beta)^{2}r^{*}-\phi(\beta)q^{*}+s^{*}]dy^{*}(\beta)=0\\
\Rightarrow & \int_{0}^{\beta}q^{*}(\alpha)d\phi(\alpha)=\phi(\beta)^{2}r^{*}+\phi(\beta)q^{*}-s^{*},\forall\beta\in[0,1]\;\text{(since }y^{*}(\beta)\;\text{is uniform over }[0,1])\\
\Rightarrow & \int_{0}^{\beta}q^{*}(\alpha)d\phi(\alpha)=\int_{0}^{\beta}(2\phi(\alpha)r^{*}+q^{*})d\phi(\alpha)+\phi(0)^{2}r^{*}+\phi(0)q^{*}-s^{*},\forall\beta\in[0,1]\\
\Rightarrow & \int_{0}^{\beta}(q^{*}(\alpha)-2\phi(\alpha)r^{*}-q^{*})d\phi(\alpha)=\phi(0)^{2}r^{*}+\phi(0)q^{*}-s^{*},\forall\beta\in[0,1],\\
\Rightarrow & \int_{\beta_{1}}^{\beta_{2}}(q^{*}(\alpha)-2\phi(\alpha)r^{*}-q^{*})d\phi(\alpha)=0,\forall\beta_{1},\beta_{2}\in[0,1].
\end{align*}
where the second line can also see \cite{Pullan:1996aa}, and the
third is because of applying integration by parts to the right-hand-side.
Hence, for any $\alpha\in\text{supp}(\phi)$, we must have $q^{*}(\alpha)=2\phi(\alpha)r^{*}+q^{*}$,
which is non-decreasing given that $\phi$ is so (note that $r^{*}\geq0)$.
Since for any $\alpha\in[0,1]\setminus\text{supp}(\phi)$, the change
of $q^{*}(\alpha)$ does not make a difference in \eqref{eq:con},
we thus confirmed that there exists also a non-decreasing function
that is optimal in the relaxed problem. This completes the proof. 
\end{proof}
The above result not only provides a unified perspective on generating
WCSRM in closed form for different choice of spectrum $\phi$, i.e.
modifying the scale factor for standard deviation accordingly, it
also enables one to re-interpret the earlier result of WCCVaR. While
the scale factor for standard deviation in WCCVaR has often been expressed
as $\sqrt{\frac{1-\epsilon}{\epsilon}}$, which appears to be the
square root of the ratio between the probability of non-tail part
and tail part, the above result explains that the ratio can also be
interpreted as how much more ``skewed'' the $\phi$ is, i.e. $\int_{0}^{1}\phi^{2}(p)dp$,
compared to the case where $\phi$ is uniform over $[0,1]$, in which
case $\int_{0}^{1}\phi(p)^{2}dp=1$ . We have $\rho_{WCSRM}(\mu,\sigma,\phi^{*})=\mu$
when $\phi^{*}$ is uniform. Hence, the closed-form might be roughly
read as ``the risk neutral value where $\phi$ is uniform plus standard
deviation multiplied by how much more skewed a given $\phi$ is compared
to the case of uniform measure, i.e. $\sqrt{\int\phi^{2}(p)dp-1}$''.

Despite its elegancy, the result of Theorem \ref{Any-worst-case-spectral}
is in fact not obvious and can be quite counter-intuitive if one takes
the perspective from the nominal risk measures, i.e. for a fixed distribution.
To see why the result might be surprising, let us highlight the following
implication from the result. 
\begin{cor}
\label{Cor1:The-worst-case-spectral}Given any $(\mu,\sigma)$, the
worst-case spectral risk measure with spectrum $\phi$, i.e. $\rho_{WCSRM}(\mu,\sigma,\phi)$
is equivalent to the worst-case $(1-\epsilon')-$value-at-risk and
$(1-\epsilon')-$conditional value-at-risk respectively with $\epsilon'=\frac{1}{\int_{0}^{1}\phi(p)^{2}dp}$,
i.e. $\rho_{WCVaR}(\mu,\sigma,\frac{1}{\int_{0}^{1}\phi(p)^{2}dp})$
and $\rho_{WCCVaR}(\mu,\sigma,\frac{1}{\int_{0}^{1}\phi(p)^{2}dp})$
respectively. 
\end{cor}
Obviously, the above statement might not be true for the case of nominal
risk measures, since in order to match the value of a spectral risk
measure, the corresponding tail probability of a $(1-\epsilon')-$CVaR
could depend on the shape of the given distribution, i.e. $\epsilon':=\epsilon(F_{Z})$.
In the above corollary however, the equivalency between WCSRM and
WCCVaR can be established independently from the structure of the
distribution, i.e. the mean and standard deviation.

We are now ready to present the result for the general case of worst-case
law invariant coherent risk measures, which can be straightforwardly
obtained from the result of Theorem \ref{Any-worst-case-spectral}
. 
\begin{thm}
\label{Thm3:Any-worst-case-law}Given Assumption 1, any worst-case
law invariant coherent risk measure defined based on $\rho_{\Phi}=\sup_{\phi\in\Phi}\rho_{\phi}$
can be evaluated in closed-form

\[
\rho_{WCLICRM}(\mu,\sigma,\Phi)=\mu+\sigma\sqrt{\sup_{\phi\in\Phi}\int_{0}^{1}\phi^{2}(p)dp-1},
\]
and is equivalent to the worst-case $(1-\epsilon')-$VaR and $(1-\epsilon')-$CVaR
by setting $\epsilon'=\frac{1}{\sup_{\phi\in\Phi}\int_{0}^{1}\phi^{2}(p)dp}$. 
\end{thm}
\begin{proof}
For simplicity, we write $F_{Z}\sim(\mu,\sigma)$ to denote any distribution
with mean $\mu$ and standard deviation $\sigma$. We have 
\begin{align*}
\rho_{WCLICRM}(\mu,\sigma,\Phi):= & \sup_{F_{Z}\sim(\mu,\sigma)}\sup_{\phi\in\Phi}\rho_{\phi}(Z)\\
= & \sup_{\phi\in\Phi}\sup_{F_{Z}\sim(\mu,\sigma)}\min_{q(\alpha)}\int[\phi(0)z+\int_{0}^{1}[(1-\alpha)q(\alpha)+(z-q(\alpha))^{+}]d\phi(\alpha)]dF_{Z}\\
= & \sup_{\phi\in\Phi}\mu+\sigma\sqrt{\int_{0}^{1}\phi^{2}(p)dp-1}\\
= & \mu+\sigma\sqrt{\sup_{\phi\in\Phi}\int_{0}^{1}\phi^{2}(p)dp-1},
\end{align*}
where the last equality is simply due to the fact that $\sigma\sqrt{\cdot}$
is an increasing function. 
\end{proof}
We end this section by concluding that the closed-form insight from
WCVaR and WWCVaR can be well carried over to many risk measures that
are considered sensible in the modern risk theory.

\section{Robust Portfolio Optimization}

The observation made in Corollary \ref{Cor1:The-worst-case-spectral}
(or Theorem \ref{Thm3:Any-worst-case-law}) can be found particularly
useful in the context of robust portfolio optimization. We provide
necessary details in this section to draw the connection between robust
portfolio optimization for VaR and general law invariant risk measures.
The problem of robust portfolio optimization seeks a portfolio that
minimizes the worst-case risk while satisfying a number of constraints
such as no short-sale requirement. It can be generally formulated
as the following minmax problem when a law invariant coherent risk
measure is employed 
\begin{eqnarray}
\min_{x\in{\cal X}} & \sup_{F_{R}} & \rho_{\Phi}(-(R)^{\top}x)\label{eq:robustpro}\\
 & {\rm subject\;to} & \mathbb{E}[R]=\mu\nonumber \\
 &  & \mathbb{COV}[R]=\Sigma,\nonumber 
\end{eqnarray}
where ${\cal X}\subset\Re^{n}$ denotes a set of admissible portfolio
allocation vectors over $n$ different assets, and $R:(\Omega,{\cal F},\mathbb{P})\rightarrow\Re^{n}$
stands for the vector of random returns of the $n$ assets with its
distribution $F_{R}$. The set ${\cal X}$ is assumed to be a bounded
polytope that does not contain 0. In the above formulation, we assume
that only the mean $\mu\in\Re^{n}$ and covariance $\Sigma\in\Re^{n\times n}$
of the joint distribution of returns are known, and a portfolio $x\in{\cal X}$
is sought that minimizes the worst-case risk over the set of multivariate
distributions $F_{R}$ having $\mu$ and $\Sigma$ as mean and covariance.
The above problem appears to be difficult due to its minmax form of
objective function and the high dimensionality of the random returns.
Fortunately, we can apply the following result first to simplify the
robust problem. 
\begin{lem}
(\cite{Chen:2011aa}) Let $A:=\{a^{\top}\xi\;|\;\mathbb{E}[\xi]=\mu,\;\mathbb{COV}[\xi]=\Sigma\}$,
$B:=\{\eta\;|\;\mathbb{E}[\eta]=a^{\top}\mu,\;\mathbb{VAR}[\eta]=a^{\top}\Sigma a\}$.
For any $a\neq0\in\Re^{n}$, it holds that $A=B$. 
\end{lem}
In other words, we can equivalently reformulated the above problem
as 
\begin{eqnarray*}
\min_{x\in{\cal X}}\sup_{F_{Z}} & \rho_{\Phi}(-Z)\\
\text{subject to} & \mathbb{E}[Z]=\mu^{\top}x\\
 & \mathbb{STD}[Z]=\sqrt{x^{\top}\Sigma x}
\end{eqnarray*}
where for any fixed $x$, the random variable $Z$ is simply a random
variable with a univariate distribution $F_{Z}$. The inner maximization
problem can now be reformulated using the result of Theorem \ref{Thm3:Any-worst-case-law}
and the whole problem can be reduced to the following minimization
problem

\[
\min_{x\in{\cal X}}-\mu^{\top}x+\sqrt{x^{\top}\Sigma x}\sqrt{\sup_{\phi\in\Phi}\int_{0}^{1}\phi(p)^{2}dp-1}.
\]

Provided that the term $\sup_{\phi\in\Phi}\int_{0}^{1}\phi(p)^{2}dp$
can be solved offline, this final problem can be solved easily by
a SOCP solver \cite{Alizadeh:2003aa}. Moreover, it is identical to
the robust portfolio optimization for VaR (\cite{Laurent-El-Ghaoui:2003aa})
except the scale factor, which confirms the following fact aligned
with the observation made in Corollary \ref{Cor1:The-worst-case-spectral}. 
\begin{cor}
Given Assumption 1, solving robust portfolio optimization problem
with law invariant coherent risk measure is equivalent to solving
the robust problem with $(1-\epsilon')-$VaR (or $(1-\epsilon')-$CVaR)
with $\epsilon'=\frac{1}{\sup_{\phi\in\Phi}\int_{0}^{1}\phi^{2}(p)dp}.$ 
\end{cor}
The above fact immediately implies that one can easily extend the
above robust portfolio problem (\ref{eq:robustpro}) to the case where
the first two moments are uncertain, which has been well addressed
in the literature of robust VaR and CVaR optimization \cite{Laurent-El-Ghaoui:2003aa,Chen:2011aa,Natarajan:2010aa}\footnote{and robust mean-variance optimization.}.
To demonstrate the idea, we provide below a few possible extensions
to deal with moment uncertainty based on the work of El Ghaoui et
al. (2003) \cite{Laurent-El-Ghaoui:2003aa}, where the resulting formulations
can often be recast as conic programs \cite{NEMIROVSKI:2007aa}. We
skip the proofs since they can be found in El Ghaoui et al. (2003)
once Corollary 2 is applied. 
\begin{cor}
(c.f.\cite{Laurent-El-Ghaoui:2003aa} Section 2.2-2.4, 3.1) Given
Assumption 1, if the mean $\mu$ and covariance $\Sigma$ of the distribution
$F_{R}$ of random returns $R$ are only known to belong to a convex
set ${\cal C}\subseteq\Re^{n}\times\Re^{n\times n}$, the robust portfolio
optimization problem (\ref{eq:robustpro}) with a law invariant coherent
risk measure $\rho_{\Phi}:=\sup_{\phi\in\Phi}\rho_{\phi}$ can be
solved by the following minmax problem 
\begin{align*}
\min_{x\in{\cal X}}\max_{r,\mu,\Sigma} & -r^{\top}x\\
{\rm {\rm subject\;to\;}} & \begin{bmatrix}\Sigma & r-\mu\\
(r-\mu)^{\top} & \sup_{\phi\in\Phi}\int_{0}^{1}\phi(p)^{2}dp-1
\end{bmatrix}\succeq0,\\
 & (\mu,\Sigma)\in{\cal C}
\end{align*}
where $\succeq0$ stands for that the left-hand-side matrix is positive
semi-definite. The problem further reduces to conic programs \cite{NEMIROVSKI:2007aa}
for the following special cases

1) (Polytopic uncertainty) ${\cal C}:={\bf Co}\{(\mu_{k},\Sigma_{k})\}_{k=1,...,K}$,
where ${\bf Co}$ is the convex hull operator

2) (Componentwise bounds) ${\cal C}:=\{(\mu,\Sigma)\;|\;\mu_{L}\leq\mu\leq\mu_{U},\Sigma_{L}\leq\Sigma\leq\Sigma_{U}\}$,

3) (Uncertainty in factor's model) ${\cal C}:=\{(\mu,\Sigma)\;|\;\exists(\mu_{f},S)\;\mu=A\mu_{f},\;\Sigma=D+ASA^{\top},\;\mu_{fL}\leq\mu_{f}\leq\mu_{fU},S_{L}\leq S\leq S_{U}\}$,
where a factor model $R=Af+u$ is assumed for the random returns and
$u$ are residuals with diagonal covariance matrix $D$. 
\end{cor}
Finally, it is worth pointing out also that based on the result of
\cite{Chen:2011aa}, the robust portfolio optimization problem (\ref{eq:robustpro})
can also be solved in closed form when the feasible set ${\cal X}$
is described by a simple budget constraint, i.e. $1^{\top}x=1$. Interested
readers are referred to Theorem 2.9 in \cite{Chen:2011aa}.

\section{Concluding Remarks}

In this paper, we showed that closed-form solutions also exist for
a general class of worst-case risk measures defined based on law invariant
coherent risk measures. The result generalizes to a great extent the
existing closed-form result of worst-case Value-at-Risk and worst-case
Conditional Value-at-Risk, which have received a considerable amount
of attention in the past decade. The closed-form solutions for the
general class of measures are remarkably similar to that of VaR and
CVaR, and thus are immediately applicable in many settings where worst-case
VaR and CVaR have been implemented.

 \bibliographystyle{siam}
\bibliography{refs2}

\begin{thebibliography}{10}

\bibitem{Acerbi:smmrcrsra}
{\sc C.~Acerbi}, {\em {Spectral measures of risk: A coherent representation of
  subjective risk aversion}}, Journal of Banking \& Finance, 26 (2002),
  pp.~1505--1518.

\bibitem{Acerbi:2002aa}
{\sc C.~Acerbi and P.~Simonetti}, {\em Portfolio optimization with spectral
  measures of risk}, tech. rep., 2002.

\bibitem{Alizadeh:2003aa}
{\sc F.~Alizadeh and D.~Goldfarb}, {\em Second-order cone programming},
  Mathematical Programming, 95 (2003), pp.~3--51.

\bibitem{artzner:coherentRM}
{\sc P.~Artzner, F.~Delbaen, J.~M. Eber, and D.~Heath}, {\em {Coherent Measures
  of Risk}}, Mathematical Finance, 9 (1999).

\bibitem{Calafiore:2007aa}
{\sc G.~C. Calafiore}, {\em Ambiguous risk measures and optimal robust
  portfolios}, SIAM Journal on Optimization, 18 (2007), pp.~853--877.

\bibitem{Chen:2011aa}
{\sc L.~Chen, S.~He, and S.~Zhang}, {\em Tight bounds for some risk measures,
  with applications to robust portfolio selection}, Oper. Res., 59 (2011),
  pp.~847--865.

\bibitem{Delage:2010aa}
{\sc E.~Delage and Y.~Ye.}, {\em Distributionally robust optimization under
  moment uncertainty with application to data-driven problems}, Operations
  Research, 58 (2010), pp.~596--612.

\bibitem{follmer02:cvxRiskMeas}
{\sc H.~F\"{o}llmer and A.~Schied}, {\em Convex measures of risk and trading
  constraints}, Finance and Stochastics, 6 (2002), pp.~429--447.

\bibitem{Laurent-El-Ghaoui:2003aa}
{\sc L.~E. Ghaoui, M.~Oks, and F.~Oustry}, {\em Worst-case value-at-risk and
  robust portfolio optimization: A conic programming approach}, Operations
  Research, 51 (2003), pp.~543--556.

\bibitem{kusuoka:licrm}
{\sc S.~Kusuoka}, {\em On law invariant coherent risk measures}, in Advances in
  Mathematical Economics, S.~Kusuoka and T.~Maruyama, eds., vol.~3 of Advances
  in Mathematical Economics, Springer Japan, 2001, pp.~83--95.

\bibitem{Natarajan:2010aa}
{\sc K.~Natarajan, M.~Sim, and J.~Uichanco}, {\em Tractable robust expected
  utility and risk models for portfolio optimization}, Mathematical Finance, 20
  (2010), pp.~695--731.

\bibitem{NEMIROVSKI:2007aa}
{\sc A.~Nemirovski}, {\em Advances in convex optimization: Conic programming},
  in Proceedings of the International Congress of Mathematicians, vol.~I,
  European Mathematical Society, Z{\"u}rich, 2007, pp.~413--444.

\bibitem{Pflug:2006aa}
{\sc G.~C. Pflug}, {\em On distortion functionals}, Statistics and Risk
  Modeling (formerly: Statistics and Decisions), 24 (2006), pp.~45--60.

\bibitem{Pflug:2015aa}
{\sc G.~C. Pflug and A.~Pichler}, {\em Time consistent decisions and temporal
  decomposition of coherent risk functionals}, Mathematics of Operations
  Research, 41 (2015), pp.~682--699.

\bibitem{Pichler:2013aa}
{\sc A.~Pichler}, {\em The natural banach space for version independent risk
  measures}, Insurance: Mathematics and Economics, 53 (2013), pp.~405--415.

\bibitem{Pichler:2015aa}
\leavevmode\vrule height 2pt depth -1.6pt width 23pt, {\em Premiums and
  reserves, adjusted by distortions}, Scandinavian Actuarial Journal,  (2015),
  pp.~332--351.

\bibitem{Pullan:1996aa}
{\sc M.~C. Pullan}, {\em A duality theory for separated continuous linear
  programs}, SIAM J. Control Optim., 34 (1996), pp.~931--965.

\bibitem{Shapiro:2001aa}
{\sc A.~Shapiro}, {\em On duality theory of conic linear problems}, in
  Semi-Infinite Programming, Kluwer Academic Publishers, 2001, pp.~135--165.

\bibitem{Shapiro:2013aa}
{\sc A.~Shapiro}, {\em On kusuoka representation of law invariant risk
  measures}, Mathematics of Operations Research, 38 (2013), pp.~142--152.

\bibitem{Sion:1958aa}
{\sc M.~Sion}, {\em On general minimax theorems}, Pacific Journal of
  Mathematics, 8 (1958), pp.~171--176.

\bibitem{Wiesemann:2014aa}
{\sc W.~Wiesemann, D.~Kuhn, and M.~Sim}, {\em Distributionally robust convex
  optimization}, Operations Research, 62 (2014), pp.~1358--1376.

\bibitem{Wozabal:2014aa}
{\sc D.~Wozabal}, {\em Robustifying convex risk measures for linear portfolios:
  A nonparametric approach}, Operations Research, 62 (2014), pp.~1302--1315.

\bibitem{Zhu:2009aa}
{\sc S.~Zhu and M.~Fukushima}, {\em Worst-case conditional value-at-risk with
  application to robust portfolio management}, Oper. Res., 57 (2009),
  pp.~1155--1168.

\end{thebibliography}

\end{document}